\documentclass[10pt,a4paper]{article}
\usepackage[utf8]{inputenc}
\usepackage[T1]{fontenc}
\usepackage[english]{babel}
\usepackage{hyperref}
\usepackage{amsmath,amsthm,enumerate}
\usepackage{amssymb}
\usepackage{tikz}
\usepackage[hmargin=3cm,vmargin=3cm]{geometry}
\usepackage{mathrsfs}
\usepackage[color=green!30]{todonotes}
\usepackage{xspace}

\newtheorem{theorem}{Theorem}

\newtheorem{proposition}[theorem]{Proposition}
\newtheorem{lemma}[theorem]{Lemma}
\newtheorem{observation}[theorem]{Observation}

\newtheorem{claim}{Claim}

\newcommand{\smallqed}{\hfill{\tiny $\left(\Box\right)$}}
\newcommand{\claimproof}{\noindent\emph{Proof of claim.} }

\newcommand{\DP}{DP}
\newcommand{\MD}{MD}

\newcommand{\DS}{\gamma}

\newcommand{\decisionpb}[4]{
        \begin{minipage}{#4\textwidth}
                #1\\
                \emph{Instance:} #2\\ 
                \emph{Question:} #3
        \end{minipage}\vspace{\baselineskip}
}

\newcommand{\optpb}[4]{
        \noindent\begin{minipage}{#4\textwidth}
                #1\\
                \emph{Instance:} #2\\ 
                \emph{Task:} #3
        \end{minipage}\vspace{\baselineskip}
}

\newcommand{\DPdec}{\textsc{Detection Pair}\xspace}
\newcommand{\DPopt}{\textsc{Opt-Detection Pair}\xspace}
\newcommand{\DSdec}{\textsc{Dominating Set}\xspace}
\newcommand{\DSopt}{\textsc{Opt-Dominating Set}\xspace}
\newcommand{\MDdec}{\textsc{Metric Dimension}\xspace}
\newcommand{\MDopt}{\textsc{Opt-Metric Dimension}\xspace}
\newcommand{\SCdec}{\textsc{Set Cover}\xspace}
\newcommand{\SCopt}{\textsc{Opt-Set Cover}\xspace}

\begin{document}

\title{Parameterized and approximation complexity of the detection pair problem in graphs}
\author{Florent Foucaud\footnote{\noindent LIMOS - CNRS UMR 6158 - Universit\'e Blaise Pascal, Clermont-Ferrand (France). E-mail: florent.foucaud@gmail.com}
\and Ralf Klasing\footnote{\noindent LaBRI - CNRS UMR 5800 - Université de Bordeaux, F-33400 Talence (France). E-mail: ralf.klasing@labri.fr}
}

\maketitle

\begin{abstract}
We study the complexity of the problem \DPdec. A detection pair of a graph $G$ is a pair $(W,L)$ of sets of detectors with $W\subseteq V(G)$, the watchers, and $L\subseteq V(G)$, the listeners, such that for every pair $u,v$  of vertices that are not dominated by a watcher of $W$, there is a listener of $L$ whose distances to $u$ and to $v$ are different. The goal is to minimize $|W|+|L|$. This problem generalizes the two classic problems \DSdec and \MDdec, that correspond to the restrictions $L=\emptyset$ and $W=\emptyset$, respectively. \DPdec was recently introduced by Finbow, Hartnell and Young [A. S. Finbow, B. L. Hartnell and J. R. Young. The complexity of monitoring  a network with both watchers and listeners. \emph{Networks}, accepted], who proved it to be NP-complete on trees, a surprising result given that both \DSdec and \MDdec are known to be linear-time solvable on trees. It follows from an existing reduction by Hartung and Nichterlein for \MDdec that even on bipartite subcubic graphs of arbitrarily large girth, \DPdec is NP-hard to approximate within a sub-logarithmic factor and W[2]-hard (when parameterized by solution size). We show, using a reduction to \SCdec, that \DPdec is approximable within a factor logarithmic in the number of vertices of the input graph. Our two main results are a linear-time $2$-approximation algorithm and an FPT algorithm for \DPdec on trees.

\vspace{0.5cm}
\end{abstract}

\section{Introduction}

In order to monitor faults in a network, one can place detectors on its nodes. One possibility is to use ``local'' devices that are able to detect the location of a fault within distance one (we call them \emph{watchers}). Another kind of detectors, that are more far-reaching, are able to determine the exact distance of a fault to the device, but not its precise location (we call such detectors \emph{listeners}). If we wish to monitor a network by using only watchers, we have the classic problem \DSdec~\cite[GT2]{GJ79} (see the books~\cite{HHS98,dom:advancedtopics} for a survey of results on this problem). If, on the other hand, we want to use only listeners, we have the problem \MDdec~\cite[GT61]{GJ79} (see for example the papers~\cite{BC,CEJO00,HM76,HN13,KRR96,S75} and references therein). However, it can be useful to use both kinds of detectors: for example, one watcher is enough to monitor a complete graph of order~$n$, while we would need $n-1$ listeners to do so. On the other hand, one listener would suffice to monitor a path of order~$n$, but for the same task we would need $\lceil n/3\rceil$ watchers. Therefore, Finbow, Hartnell and Young~\cite{FHY15} recently proposed the concept of a \emph{detection pair} in a graph $G$, that is, a pair $(W,L)$ where $W\subseteq V(G)$ is a set of watchers and $L\subseteq V(G)$ is a set of listeners that, together, monitor the graph $G$.

More formally, we say that a vertex $w$ \emph{dominates} all the vertices in the closed neighbourhood $N[w]$ of $w$ (that is, $N[w]$ is the set of neighbours of $w$ together with $w$ itself). Moreover, a vertex $l$ \emph{separates} two vertices $u$ and $v$ if the distance $d(u,l)$ is different from the distance $d(v,l)$. Given a pair $(W,L)$ of sets with $W\subseteq V(G)$ (the \emph{watchers}) and $L\subseteq V(G)$ (the \emph{listeners}), we say that a vertex $u$ is \emph{distinguished} by $(W,L)$ either if $u$ is dominated by a watcher of $W$, or if for every other vertex $v$, either $v$ is dominated by a watcher of $W$, or $u$ and $v$ are separated by a listener of $L$. Then, $(W,L)$ is a \emph{detection pair} of $G$ if every vertex of $G$ is distinguished by $(W,L)$. The size of $(W,L)$ is the sum $|W|+|L|$ and is denoted by $||(W,L)||$. Note that we may have $W\cap L\neq \emptyset$ if we choose to place both a listener and a watcher at the same position. We denote by $\DP(G)$ the minimum size of a detection pair of $G$. Any vertex of $(W,L)$ is called a \emph{detector}.

If we have a detection pair $(\emptyset,L)$ in a graph $G$, then $L$ is called a \emph{resolving set} of $G$, and the smallest size of a resolving set of $G$ is called its \emph{metric dimension}, denoted $\MD(G)$. On the other hand, if we have a detection pair $(W,\emptyset)$, then $W$ is a \emph{dominating set} of $G$, and the smallest size of a dominating set of $G$ is its \emph{domination number}, denoted $\DS(G)$. Clearly, we have $\DP(G)\leq\min\{\DS(G),\MD(G)\}$, and the inequality can be strict~\cite{FHY15}. It follows that for any graph $G$ without isolated vertices, we have $\DP(G)\leq \DS(G)\leq |V(G)|/2$~\cite{HHS98}.

The goal of this paper is to study the decision and optimization problems that are naturally associated to the notion of a detection pair (other computational problems that are mentioned in this paper are defined formally in Section~\ref{sec:prelim}).

\bigskip

\noindent\decisionpb{\DPdec}{A graph $G$, a positive integer $k$.}{Is there a detection pair $(W,L)$ of $G$ with $||(W,L)||\leq k$?}{1}

\noindent\optpb{\DPopt}{A graph $G$.}{Compute an optimal detection pair of $G$.}{1}

A \emph{$c$-approximation algorithm} for a given optimization problem $\Pi_O$ is an algorithm that returns a solution whose size is always at most $c$ times the optimum. We refer to the book~\cite{ACGKMP99} for more details.
For a decision problem $\Pi$ and for some parameter $p$ of the instance, an algorithm for $\Pi$ is said to be \emph{fixed parameter tractable} (FPT for short) if it runs in time $f(p)n^c$, where $f$ is a computable function, $n$ is the input size, and $c$ is a constant. In this paper we will always consider the solution size $k$ as the parameter. We refer to the books~\cite{DF13,N06} for more details.

\bigskip

Finbow, Hartnell and Young proved that \DPdec is NP-complete on trees~\cite{FHY15}, while \DPopt is linear-time solvable on trees containing no pair of leaves with a common neighbour. This hardness result is quite surprising since the related problems \DSopt and \MDopt, while being NP-hard in general~\cite[GT2,GT61]{GJ79}, can be solved in linear time on trees (see respectively~\cite{CGH75} and~\cite{HM76,S75}). We note that \DSdec is among the most classic and well-studied graph problems (see the books~\cite{HHS98,dom:advancedtopics}), and that \MDdec has enjoyed a lot of interest in the recent years: see the papers~\cite{MFCS,DPSV12,Eppstein,ELW12j,FHvMS15,foucaudWG,HN13,HSV12}.

\bigskip

In this paper, we continue the study of the complexity of \DPdec and \DPopt initiated in~\cite{FHY15}. In Section~\ref{sec:gen}, we describe a reduction to \SCopt that shows that \DPopt can be solved within a factor logarithmic in the size of the input graph. On the other hand, we observe that a reduction of Hartung and Nichterlein~\cite{seppThesis,HN13} for \MDopt can also be applied to \DPopt. This implies that \DPopt is NP-hard to approximate within a factor that is sublogarithmic in the input graph's order, and that \DPdec is W[2]-hard when parameterized by the solution size $k$. (These hardness results hold even for graphs that are bipartite, subcubic, and have arbitrarily large girth.) In Section~\ref{sec:2approx-trees}, we prove that \DPopt is $2$-approximable in linear time on trees. In Section~\ref{sec:fpt-trees}, we show that there is an algorithm for \DPdec running in FPT time $2^{O(k\log k)}n^2$ on trees.\footnote{In this paper, ``$\log$'' denotes the natural logarithm function.} We start the paper with some preliminary considerations in Section~\ref{sec:prelim}, and conclude in Section~\ref{sec:conclu}.

\section{Preliminaries}\label{sec:prelim}

We start with some preliminary considerations.

\subsection{Definitions of used computational problems}

We now formally define a few auxiliary computational problems that are used or mentioned in this paper.

\bigskip

\noindent\decisionpb{\DSdec \cite[GT2]{GJ79}}{A graph $G$, a positive integer $k$.}{Do we have $\DS(G)\leq k$, that is, is there a dominating set $D$ of $G$ with $|D|\leq k$?}{1}

\noindent\optpb{\DSopt}{A graph $G$.}{Compute an optimal dominating set of $G$.}{1}

\noindent\decisionpb{\MDdec \cite[GT61]{GJ79}}{A graph $G$, a positive integer $k$.}{Do we have $\MD(G)\leq k$, that is, is there a resolving set $R$ of $G$ with $|R|\leq k$?}{1}

\noindent\optpb{\MDopt}{A graph $G$.}{Compute an optimal resolving set of $G$.}{1}

\noindent\decisionpb{\SCdec \cite[SP5]{GJ79}}{A hypergraph $H=(X,S)$, a positive integer $k$.}{Is there a set cover $C$ of $H$ with $|C|\leq k$ (that is, $C\subseteq S$ and each vertex of $X$ belongs to some set of $C$)?}{1}

\noindent\optpb{\SCopt}{A hypergraph $H$.}{Compute an optimal set cover of $H$.}{1}

\subsection{Specific terminology}

In a graph $G$, a vertex of degree~$1$ is called a \emph{leaf}. A vertex is called a \emph{branching point} if its degree is at least~$3$. A branching point $v$ is \emph{special} if there is a path $L$ starting at $v$ and ending at a vertex of degree~$1$, and whose inner-vertices all have degree~$2$. Path $L-v$ is called a \emph{leg} of $v$, and we say that $L$ is \emph{attached} to $v$ (note that vertex $v$ does not belong to the leg). A special branching point with $t$ leaves as neighbours is called a \emph{$t$-stem}.\footnote{The previous terminology is from~\cite{FHY15} (branching points, stems) and~\cite{KRR96} (legs).}

Given a special branching point $x$ of a tree $T$, we define the subtree $L(x)$ of $T$ as the tree containing $x$ and all legs of $T$ attached to $x$.

\subsection{The classic algorithm for \MDopt on trees}

We will use the following results of Slater~\cite{S75} about \MDopt on trees, that are classic in the literature about metric dimension. See also~\cite{HM76,KRR96} for similar considerations.

\begin{proposition}[Slater \cite{S75}]\label{prop:MD-trees}
Let $T$ be a tree and $R$ be the set of vertices containing, for each special branching point $x$ of $T$ that has $\ell$ legs attached, the leaves of the $\ell-1$ longest of these legs. Then, $R$ is an optimal resolving set of $T$.
\end{proposition}

We have the following consequence of Proposition~\ref{prop:MD-trees}.

\begin{theorem}[Slater \cite{S75}]\label{thm:slater}
\MDopt can be solved optimally in linear time on trees.
\end{theorem}

\subsection{Observations and lemmas about detection pairs}

The following easy observations and lemmas will be useful.

\begin{observation}\label{obs:specialBP}
Let $G$ be a graph and $B_{2^+}$ the set of its special branching points with at least two legs attached. Then, we have $DP(G)\geq |B_{2^+}|$.
\end{observation}
\begin{proof}
Let $(W,L)$ be a detection pair of $G$. We show that for each special branching point $x$ of $B_{2^+}$, there is at least one detector of $(W,L)$ among $x$ and the vertices belonging to a leg attached to $x$. Indeed, if not, then any two neighbours of $x$ belonging to this set would neither be dominated by a watcher, nor separated by a listener, a contradiction.
\end{proof}

\begin{lemma}\label{lem:legs}
Let $x$ be a special branching point of a graph $G$ with a set $S_x$ of $t$ leaf neighbours and a set $\mathcal L_x$ of $\ell$ legs of length at least~$2$ attached to $x$. Let $(W,L)$ be a detection pair of $G$. If there is a watcher at $x$, at least $\ell-1$ of the legs in $\mathcal L_x$ contain a detector of $(W,L)$. Otherwise, at least $t+\ell-1$ legs attached to $x$ contain a detector. 
\end{lemma}
\begin{proof}
If there is a watcher at $x$, assume for a contradiction that two legs $L_1$ and $L_2$ of $\mathcal L_x$ contain no detector. Then, neither the listeners of $(W,L)$ nor the watcher at $x$ can distinguish the two vertices of $L_1$ and $L_2$ that are at distance~$2$ from $x$, a contradiction. A similar argument holds for vertices at distance~$1$ of $x$ when $x$ has no watcher.
\end{proof}

\begin{lemma}\label{lem:unique-leg}
Let $G$ be a graph with a special branching point $x$ having a leg attached, whose leaf is $y$. If $G$ has a detection pair $(W,L)$ with $x\in L$, then $(W,L\setminus\{x\}\cup\{y\})$ is also a detection pair of $G$.
\end{lemma}
\begin{proof}
Let $V_x$ be the set of vertices containing $x$ and the vertices of the leg attached to $x$ whose leaf is $y$. Clearly, a pair $u,v$ of vertices in $V(G)\setminus V_x$ are separated by $x$ if and only if they are separated by $y$. Moreover, all the vertices in the leg containing $y$ are clearly distinguished, because each such vertex is the unique one with its distance to $y$. This completes the proof.
\end{proof}

\section{General approximability and non-approximability}\label{sec:gen}

In this section, we discuss the general approximation complexity of \DPopt.

\begin{theorem}
\DPopt can be approximated within a factor $2\log n+1$ on graphs of order $n$.
\end{theorem}
\begin{proof}
Given a graph $G$, we build an instance $(X,S)$ of \SCopt as follows. We let $X$ be the set of vertex pairs of $G$. For each vertex $v$ of $G$, we have a set $W_v$ of $S$ that contains an element $\{x,y\}$ of $X$ if $v$ dominates at least one of $x$ and $y$. We also have a set $L_v$ that contains an element $\{x,y\}$ of $X$ if we have $d(v,x)\neq d(v,y)$. Therefore we have $|X|=\binom{|V(G)|}{2}$ and  $|S|=2|V(G)|$. Now, we claim that we have a one-to-one correspondance between the detection pairs of $G$ and the set covers of $(X,S)$. Indeed, for every vertex $v$ of $G$, the set $W_v$ of $S$ corresponds to a watcher placed at vertex $v$, and the set $L_v$ corresponds to a listener placed at vertex $v$. Moreover set $U_v$ of $S$ (with $U\in\{W,L\}$) covers exactly the elements of $X$ that correspond to a pair that is separated by the according detector placed on $v$ in $G$. Since \SCopt is approximable within a factor of $\log(|X|)+1$ in polynomial time~\cite{J74}, the result follows.
\end{proof}

Hartung and Nichterlein~\cite{HN13} provided a reduction from \DSdec to \MDdec. This reduction was improved by Hartung in his PhD thesis~\cite{seppThesis} to get a reduction mapping any instance $(G,k)$ to an instance $(G',k+4)$, where $G'$ is a bipartite graph of maximum degree~$3$ and has girth $4|V(G)|+6$. In fact, it is not difficult to see that for any detection pair $P$ of $G'$ of size $k'$, there is a detection pair $P'$ of $G'$ of size at most $k'$ containing no watchers. In other words, we have $\DP(G')=\MD(G')$, and therefore the aforementioned reduction is also a reduction to \DPdec. Hence we have the following.\footnote{Hartung and Nichterlein only state the hardness of approxmation for factors in $o(\log n)$, but a recent result on the inapproximability of \SCopt~\cite{DS14}, that transfers to \DSopt via standard reductions, implies our stronger statement.}

\begin{theorem}[Hartung and Nichterlein \cite{seppThesis,HN13}]
\DPopt is NP-hard to approximate within a factor of $(1-\epsilon)\log n$ (for any $\epsilon>0$) for instances of order~$n$, and \DPdec is W[2]-hard (parameterized by the solution size $k$). Let $g\geq 4$ be an arbitrary integer. These results hold even for instances that are bipartite, have maximum degree~$3$, and girth at least~$g$.
\end{theorem}

\section{A $2$-approximation algorithm on trees}\label{sec:2approx-trees}

In this section, we prove the following approximability result.

\def\S2{S_{2}}
\def\S3{S_{3^+}}

\begin{theorem}
There is a linear-time $2$-approximation algorithm for \DPopt on trees.
\end{theorem}
\begin{proof}
Given an input tree $T$, we denote by $\S3$, the set of vertices of $T$ that are $t$-stems for some $t\geq 3$. Moreover, we define $T'$ as the tree obtained from $T$ as follows. For each vertex $v$ of $\S3$ (that is, $v$ is a $t$-stem for some $t\geq 3$), assuming that there are $\ell$ legs of length at least~$2$ attached to $v$ in $T$, we remove from $T$: (i) $t-2$ leaf neighbours of $v$ if $\ell=0$, (ii) $t-1$ leaf neighbours of $v$ if $\ell=1$, and (iii) all $t$ leaf neighbours of $v$ if $\ell\geq 2$. Therefore, in $T'$, there are at least two legs attached to each vertex of $\S3$.

Let us describe our algorithm $A$.

\begin{itemize}
\item[1.] Compute the set $\S3$ and build $T'$ from $T$.
\item[2.] Compute an optimal resolving set $R'$ of $T'$.
\item[3.] Output $A(T)=(\S3,R')$.
\end{itemize}

We first prove that $A(T)$ is a detection pair of $T$. Notice that the distances from listeners in $R'$ and vertices of $V(T')$ are the same in $T$ and $T'$. Therefore, all pairs of vertices of $V(T')$ are distinguished in $T$ by the listeners of $A(T)$. On the other hand, all vertices of $V(T)-V(T')$ are dominated by the watchers of $\S3$. Therefore, $A(T)$ is a detection pair of $T$.

Moreover, it is clear that $\S3$ and $T'$ can be computed in linear time, and $R'$ can be computed in linear time using Slater's algorithm (Theorem~\ref{thm:slater}). Therefore $A$ is a linear-time algorithm.

It remains to prove that $||A(T)||=|\S3|+\MD(T')\leq 2\cdot\DP(T)$ for any tree $T$. Notice that $A(T)$ contains only vertices that are part of some subtree $L(x)$, where $x$ is a special branching point of $T$; indeed, any vertex of $\S3$ is a special branching point of $T$, and any special branching point of $T'$ is also a special branching point of $T$ (by Proposition~\ref{prop:MD-trees}, any optimal resolving set of $T'$ only contains vertices that belong to some subtree $L(x)$ of $T'$). Therefore, it suffices to prove that for any detection pair $D=(W,L)$ of $T$, for any optimal resolving set $R'$ of $T'$ and for every special branching point $x$ of $T$, we have $|\S3\cap V(L(x))|+|R'\cap V(L(x))|\leq 2|W\cap V(L(x))|+2|L\cap V(L(x))|$. To this end, let us assume that $x$ is a $t$-stem for some $t\geq 0$, and that there are $\ell$ legs of length at least~$2$ attached to $x$. We distinguish the following cases.

\medskip

\noindent\textbf{Case 1: \boldmath{$\ell\leq 1$}.} Then, $L(x)$ contains at most two detectors of $A(T)$. Indeed, if $t\leq 2$, it contains one or two listeners; if $t\geq 3$, it contains one watcher and one listener. By Lemma~\ref{lem:legs}, there is at least one detector in $L(x)$, which completes the proof of this case.

\medskip

\noindent\textbf{Case 2: \boldmath{$\ell\geq 2$}.} Then, by Lemma~\ref{lem:legs}, any detection pair $(W,L)$ of $T$ has a detector in at least $\ell-1$ legs of length at least~$2$ of $L(x)$. If $t\geq 3$ (that is, $x\in\S3$), $A(T)$ has exactly $\ell-1$ listeners and one watcher in $L(x)$. Therefore, for any detection pair $(W,L)$ of $T$, we have $2|W\cap V(L(x))|+2|L\cap V(L(x))|\geq 2\ell-2\geq \ell=|\S3\cap V(L(x))|+|R'\cap V(L(x))|$ (because $\ell\geq 2$), and we are done.

Hence, assume that $t\leq 2$. Then, $L(x)$ contains only listeners of $A(T)$, that is, precisely $t+\ell-1$ leaves of $L(x)$ that belong to $R'$. If $t\leq 1$, we have $\ell-1\geq t$ and then, for any detection pair $(W,L)$ of $T$, we have $2|W\cap V(L(x))|+2|L\cap V(L(x))|\geq 2\ell-2\geq t+\ell-1=|\S3\cap V(L(x))|+|R'\cap V(L(x))|$, and we are done. If however, $t=2$, by Lemma~\ref{lem:legs}, there are at least $\ell$ detectors of $(W,L)$ in $L(x)$. Hence, we have $2|W\cap V(L(x))|+2|L\cap V(L(x))|\geq 2\ell\geq t+\ell>t+\ell-1=|\S3\cap V(L(x))|+|R'\cap V(L(x))|$. This completes the proof.
\end{proof}

There are two infinite families of trees that show that our algorithm has no better approximation ratio than~$2$.

The first family consists of trees $T^1_\ell$ ($\ell\geq 1$) that are built from $\ell$ disjoint stars with at least three leaves each, all whose centers are adjacent to an additional single vertex. We have $||A(T^1_\ell)||=2\ell$ (this solution set contains one watcher and one listener for every star), while $\DP(T^1_\ell)\leq\ell$ (simply let every center of a star be a watcher, and select no listener). Hence the approximation ratio of $A$ on $T^1_\ell$ is $2$.

A second family consists of trees $T^2_\ell$ ($\ell\geq 1$) built from a path $P$ of $\ell+2$ vertices. For each of the two end-vertices of $P$, add two leaves that are adjacent to this end-vertex. For each degree~$2$-vertex $v$ of $P$, build a star $S_v$ with three leaves and subdivide one of its edges once. Moreover, the center of $S_v$ is made adjacent to $v$. We have $||A(T^2_\ell)||=2\ell+2$ (each star contains two listeners, and two additional listeners are selected among the four degree~$1$-vertices adjacent to end-vertices of $P$), while $\DP(T^2_\ell)\leq \ell+2$ (put a watcher on each center of a star, and add one neighbour of each endpoint of $P$ as a listener). The approximation ratio of $A$ on $T^2_\ell$ is $2-\tfrac{2}{\ell+2}$.

\section{A fixed parameter tractable algorithm on trees}\label{sec:fpt-trees}

In this section, we provide an exact algorithm for \DPdec that is FPT for the natural parameter $k$, the solution size.

The idea of the algorithm is as follows. We first search for a solution with $|L|\leq 1$. After this step, we may assume that $|L|\geq 2$, which is a technical condition required for the subsequent steps of the algorithm. Then we proceed in three phases. In the first phase, we handle the solution around the special branching points of the tree. As we will see, there is only a fixed set of possibilities to try for each such branching point. Then, in the second phase, we determine the set of remaining listeners. Here, we are able to compute a set (whose size is bounded in terms of $k$) of possible vertices that may contain a listener in an optimal solution. Finally, in the third phase, we determine the set of remaining watchers, and again we are able to compute a set of vertices (whose size is bounded in terms of $k$) that may be used as a watcher. The algorithm then checks the validity of each possible choice of these placements.

\begin{theorem}\label{thm:fpt}
\DPdec can be solved in time $2^{O(k\log k)}n^2$ on trees of order $n$.
\end{theorem}

\subsection{Preliminary lemmas}

Before describing the algorithm, we prove a series of lemmas that will be essential.

The first lemma is useful to handle special branching points of the tree, by reducing the problem around these vertices to a fixed number of cases.

\begin{lemma}\label{lemma:SBP}
Let $G$ be a graph with a special branching point $x$ with at least two legs and $(W,L)$ an optimal detection pair of $G$. Let $S_x$ be the set of leaves adjacent to $x$, and let $\mathcal{L}_x$ be the set of legs of length at least~$2$ attached to $x$. Denote by $V(\mathcal{L}_x)$ the vertices belonging to a leg in $\mathcal{L}_x$, and let $\ell=|\mathcal{L}_x|$ and $t=|S_x|$ (so, $x$ is a $t$-stem). Then, we can obtain from $(W,L)$ another optimal detection pair of $G$ by replacing all the detectors in $V_x=\{x\}\cup S_x\cup V(\mathcal{L}_x)$ by one of the following sets of detectors:
\begin{itemize}
\item[1.] If $\ell\leq 1$, we may use:
  \begin{itemize}
  \item[(a)] a single watcher at $x$;
  \item[(b)] a watcher at $x$, and a listener at the leaf of a longest leg attached to $x$;
  \item[(c)] if $t+\ell=2$, a single listener at the leaf of a longest leg attached to $x$.
  \end{itemize}
\item[2.] If $\ell\geq 2$, we may use:
  \begin{itemize}
  \item[(a)] a watcher at $x$ and $\ell$ listeners at all leaves of the legs of $\mathcal{L}_x$;
  \item[(b)] a watcher at $x$ and $\ell-1$ listeners at all leaves of the legs of $\mathcal{L}_x$ but the shortest one;
  \item[(c)] if $t\leq 1$, $\ell+t$ listeners at all leaves of the legs attached to $x$;
  \item[(d)] if $t\leq 1$, $\ell+t-1$ listeners at all leaves of the legs attached to $x$, except for a shortest such leg.
  \end{itemize}
\end{itemize}
\end{lemma}
\begin{proof} Clearly, if there are any watchers among $\{x\}\cup S_x$, we may replace them by a single watcher at $x$ and obtain a valid detection pair. We now distinguish two cases.

\medskip

\noindent\textbf{Case~1: \boldmath{$\ell\leq 1$}.} Then, we have $t\geq 1$. If $V_x$ contained a unique detector of $(W,L)$ that is a watcher, then it must have dominated a leaf neighbour of $x$ (since $t\geq 1$, such a neighbour exists). Then as observed at the beginning of the proof, we may replace this detector by a watcher at $x$, and we are in Case~1(a). If $V_x$ had a unique detector of $(W,L)$ that is a listener, then by Lemma~\ref{lem:legs} we must have had $t+\ell=1$. Note that any listener in $V_x$ seperates the same set of pairs of vertices from $V(G)\setminus V_x$. Hence we may replace the existing listener by a listener at the leaf of a longest leg attached to $x$, which separates at least as many vertices as any other listener in $V_x$. Then we are in Case~1(c) and we are done. Hence, we assume that $V_x$ contained at least two detectors of $(W,L)$. Then, using similar arguments one can check that the solution of Case~1(b) yields a valid detection pair of $G$ that is not larger than $(W,L)$.

\medskip

\noindent\textbf{Case~2: \boldmath{$\ell\geq 2$}.} As observed before, we can replace all watchers in $\{x\}\cup S_x$ by a single watcher at $x$. If $x$ has a watcher, we let $\mathcal L_x^+=\mathcal L_x$; otherwise, we let $\mathcal L_x^+$ be the set of all legs attached to $x$. Let $\ell^+=|\mathcal L_x^+|$.

By Lemma~\ref{lem:legs}, at least $\ell^+-1\geq 1$ of the legs in $\mathcal L_x^+$ contain a detector of $(W,L)$. Assume first that $(W,L)$ contained at least $\ell^+$ detectors of $(W,L)$ on the legs of $\mathcal L_x^+$. Then, we replace all such detectors by a listener at the leaf of each leg of $\mathcal L_x^+$. Then, any vertex of a leg $L_i$ of $\mathcal L_x^+$ is distinguished by its distance to the listener placed at the leaf of $L_i$. Since any listener in $V_x$ separates the same set of pairs in $V(G)\setminus V_x$, we have obtained a valid detection pair that is not larger than $(W,L)$ and we are in Case~2(a) or~2(c).

Suppose now that there are exactly $\ell^+-1$ detectors of $(W,L)$ on the legs of $\mathcal L_x^+$. We replace these detectors by a listener at the leaf of each of the $\ell^+-1$ longest legs of $\mathcal L_x^+$. By similar arguments as above, all vertices in $V_x$ that do not belong to the unique leg $L_1$ of $\mathcal L_x^+$ without a detector are distinguished. Again, any listener in $V_x$ separates the same set of pairs of $V(G)\setminus V_x$. Hence, if some vertex is not distinguished by the new detection pair, it must be a vertex $u$ of $L_1$ that is not separated from a vertex $v$ of $V(G)\setminus V_x$. This implies that $d(u,x)=d(v,x)$ (otherwise $u$ and $v$ would be separated by any of the listeners placed in $V_x$). We know that $u$ and $v$ were separated by some detector of $(W,L)$, which must have been a detector in the leg $L_1$ (no detector in another leg of $\mathcal{L}_x^+$ could possibly separate $u$ and $v$ since $d(u,x)=d(v,x)$). But we know that there was one leg $L_2$ of $\mathcal L_x^+$ that contained no detector of $(W,L)$. Since $L_1$ is a shortest leg of $\mathcal{L}_x^+$, there is a vertex $u'$ in $L_2$ with $d(x,u)=d(x,u')$. But then, $u'$ and $v$ were not distinguished by $(W,L)$, a contradiction. Hence we have a valid detection pair whose size is not greater than $||(W,L)||$ and we are in Case~2(b) or~2(d), and we are done.
\end{proof}

The next two lemmas are used in our algorithm to choose the placement of the listeners.

\begin{lemma}\label{lemm:between2listeners}
Let $G$ be a graph with two vertices $u$ and $v$ connected by a unique path $P$. Let $V(P^+)$ be the set of vertices of $P$ together with the vertices of each leg attached to an inner-vertex of degree~$3$ of $P$.
Then, by placing listeners on $u$ and on $v$, all vertices of $V(P^+)$ are distinguished (by $u$ and $v$).
\end{lemma}
\begin{proof}
For every vertex $x$ of $G$, we let $\delta_{u,v}(x)=d(u,x)-d(v,x)$. Note that if for two vertices $x$ and $y$, $\delta_{u,v}(x)\neq \delta_{u,v}(y)$, then $x$ and $y$ are separated by $u$ and $v$.

Clearly, $u$ and $v$ are distinguished. Let $w$ be a vertex of $V(P^+)\setminus\{u,v\}$, and assume for a contradiction that there is some vertex $z$ with $d(u,w)=d(u,z)$ and $d(v,w)=d(v,z)$. In particular, $\delta_{u,v}(w)=\delta_{u,v}(z)$.

Note that for every two vertices $p$, $q$ of $P$, we have $\delta_{u,v}(p)\neq \delta_{u,v}(q)$, therefore
at most one of the vertices $w$ and $z$ can belong to $P$. We assume w.l.o.g.~that $z$ does not belong to $P$.

The shortest paths from $z$ to both $u$ and $v$ must go through an inner-vertex of $P$, say $s$. Then, we have $\delta_{u,v}(z)=\delta_{u,v}(s)=\delta_{u,v}(w)$, which implies that all shortest paths from $w$ to $u$ and $v$ also go through $s$. But then, we have $d(w,s)=d(z,s)$, a contradiction because by the structure of $V(P^+)$ this implies that $w=z$.
\end{proof}

\begin{lemma}\label{lemm:between2listeners-coro}
Let $G$ be a graph with two vertices $u$ and $v$ connected by a unique path $P$, such that each inner-vertex of $P$ has either degree~$2$ or a single leg attached. Moreover let $V(P^+)$ be the set of inner-vertices of $P$ together with the vertices of each leg attached to an inner-vertex of degree~$3$ of $P$.
Then, for every detection pair $(W,L)$ of $G$, one can obtain a valid detection pair of $G$ by replacing all detectors of $(W\cup L)\cap V(P^+)$ by two listeners at $u$ and $v$.
\end{lemma}
\begin{proof}
By Lemma~\ref{lemm:between2listeners}, all vertices in $V(P^+)$ are distinguished by $u$ and $v$. Note that no watcher of $V(P^+)$ could possibly distinguish any vertex outside of $V(P^+)\cup\{u,v\}$ (and $u$, $v$ are clearly distinguished). Furthermore, since $P$ is the unique path connecting $u$ and $v$, for any listener $l\in L$ in $V(P^+)$, the set of pairs in $V(G)\setminus V(P^+)$ separated by $l$ is the same as the set of pairs separated by $u$ (and by $v$). This completes the proof.
\end{proof}

The next lemma will be used in the algorithm to determine the placement of watchers in parts of the tree where there is no listener.

\begin{lemma}\label{lemm:subtree-no-listener}
Let $T$ be a tree with a vertex $x$ and an optimal detection pair $(W,L)$ of $T$ with $|L|\geq 1$. Let $T_x$ be a subtree of $T$ containing $x$ and assume that (i) $T_x$ contains no $t$-stems for any $t\geq 2$, (ii) for each listener $l$ of $L$ and each vertex $v$ of $T_x$, the shortest path from $l$ to $v$ goes through $x$ (in particular $T_x-x$ contains no listener) and (iii) $T_x$ is maximal with respect to~(i) and~(ii). Then, there is a set $X$ of size at most $63|W|^2$ with $W\cap V(T_x)\subseteq X$. Moreover, $X$ can be determined in time linear in $|V(T_x)|$.
\end{lemma}
\begin{proof}
We may view $T_x$ as a rooted tree whose root is $x$. Consider the partition $\mathcal P$ of $V(T_x)$ into parts $P_0, P_1, \ldots, P_d$ where $P_i$ contains the vertices of $V(T_x)$ that have distance~$i$ to $x$. Since $|L|\geq 1$ and by~(ii) all shortest paths from any listener of $L$ to a vertex of $T_x$ go through $x$, the set $L$ of listeners separates a pair $u,v$ of $V(T_x)$ if and only if $u$ and $v$ belong to different parts of $\mathcal P$. Therefore, since $(W,L)$ is a detection pair of $G$, for any two distinct vertices $u$ and $v$ from the same part $P_i$, at least one of $u,v$ is dominated by a watcher of $W$.

Let $\mathcal S_{2^+}$ be the set of parts of $\mathcal P$ of size at least~$2$ (that is, those parts where at least one vertex must be dominated by a watcher). Let $\widetilde{F_x}$ be the forest induced by the set of vertices of $T_x$ that belong to some part of $\mathcal S_{2^+}$, together with their neighbours in $T_x$.
We let $X=V(\widetilde{F_x})$, and we shall prove that $X$ satisfies the claim. First, we show that $(W\cap V(T_x))\subseteq X$. Assume by contradiction that this is not the case, and let $w$ be a watcher of $W$ that belongs to $V(T_x)\setminus X$. By the optimality of $(W,L)$, $(W\setminus\{w\},L)$ is not a detection pair of $T$. Therefore, there must be some vertex $u$ dominated by $w$ that is not separated by $L$ from some other vertex $v$ (and $v$ is not dominated by a watcher of $W\setminus\{w\}$). In other words, for every listener $l\in L$, we have $d_T(u,l)=d_T(v,l)$. Thus, since $T$ is a tree, all shortest paths from $l$ to both $u$ and $v$ have a common part (say from $l$ until some vertex $y$) and then are disjoint. By~(iii), $T_x$ is maximal with respect to~(i) and~(ii), which implies that $y=x$ and both $u$ and $v$ belong to $T_x$. However, observe that $u$ belongs to a part of $\mathcal S_{2^+}$ (otherwise $w$ would belong to $X$). Therefore, $u$ is uniquely determined (within $T_x$) by its distance to $x$. Since $|L|\geq 1$ and by~(ii), this implies that $u$ and $v$ are in fact distinguished by $L$, a contradiction.

Thus, we have shown that $W\cap V(T_x)\subseteq X$. It is moreover clear that $X$ can be computed in time linear in $|V(T_x)|$ by using the distances of the vertices of $T_x$ to $x$. Therefore, it remains only to bound the size of $X$.

First, we show that $\widetilde{F_x}$ contains vertices of at most $9|W|$ distinct parts of $\mathcal P$. To each part $P_i$ of $\mathcal S_{2^+}$, we can associate the two parts $P_{i-1}$ and $P_{i+1}$, that may have vertices in $\widetilde{F_x}$ but that may not be in $\mathcal S_{2^+}$. Hence, the total number of parts with a vertex in $\widetilde{F_x}$ is at most $3|\mathcal S_{2^+}|$. Since a watcher may dominate vertices from at most three parts of $\mathcal S_{2^+}$ and each part in $\mathcal S_{2^+}$ has at least one vertex dominated by a watcher, we have $|\mathcal S_{2^+}|\leq 3|W|$, hence $\widetilde{F_x}$ contains vertices of at most $9|W|$ distinct parts of $\mathcal P$.

Second, we prove that $\widetilde{F_x}$ has at most $|W|+2|\mathcal S_{2^+}|\leq 7|W|$ leaves. Note that a leaf of $\widetilde{F_x}$ is either an actual leaf in $T$, or it is a vertex of a part $P_i$ with $|P_i|=1$ (or both). We claim that $\widetilde{F_x}$ can contain at most $|W|$ leaves of $T$ belonging to a part of $\mathcal S_{2^+}$. Indeed, each such leaf must be dominated by a watcher, but it can be dominated only by itself or by its parent in $T_x$. But since by~(i) there are no $t$-stems for $t\geq 2$ in $T_x$, no parent can dominate more than one leaf, which proves the bound. Furthermore, $\widetilde{F_x}$ contains vertices from at most $2|\mathcal S_{2^+}|$ parts of $\mathcal P$ of size~$1$, hence the number of leaves of $\widetilde{F_x}$ that are not in a part of $|\mathcal S_{2^+}|$ is at most $2|\mathcal S_{2^+}|\leq 6|W|$ since $|\mathcal S_{2^+}|\leq 3|W|$. Hence in total, we have at most $7|W|$ leaves in $\widetilde{F_x}$.

But any rooted forest with at most $9|W|$ levels and with at most $7|W|$ leaves may have at most $63|W|^2$ vertices, hence $X=V(\widetilde{F_x})$ has size at most $63|W|^2$. This completes the proof.
\end{proof}

\subsection{The algorithm (proof of Theorem~\ref{thm:fpt})}

We now describe the claimed FPT algorithm for \DPdec on trees and prove its correctness and running time.

\setcounter{theorem}{9}

\begin{proof}[Proof of Theorem~\ref{thm:fpt}]
  Let us describe the algorithm, that seeks to build a detection pair $(W,L)$ with at most $k$ detectors of an input tree $T$ of order $n$.
  
\medskip

\noindent\textbf{Preliminary phase: Searching for a solution with \boldmath{$|L|\leq 1$}.} A first step of the algorithm is to check whether there is a solution with $|L|\leq 1$. If we set $L=\emptyset$, the problem is equivalent to finding a dominating set of $T$ of size at most $k$. This can be done in linear time~\cite{CGH75}. If we set $|L|=1$, one may try the $n$ possibilities for the placement of the unique listener. For each possibility (say we have placed the listener at vertex $v$), we create a subtree $T_v$ of $T$ rooted at $v$ as follows. First of all, let $T_v=T$. Then, for each $t$-stem $x$ of $T$ with $t\geq 2$ (except if $v$ is a leaf attached to $x$ and $x$ is a $2$-stem), we may assume that $x$ belongs to $W$ (by a similar argument as in the proof of Lemma~\ref{lemma:SBP}, Cases~1a and~2a.). Thus, we add $x$ to $W$ and we we remove from $T_v$ all leaves but one that are attached to $x$ (if $v$ is one such leaf, we keep $v$ in $T_v$). Now $T_v$ has no $t$-stems for $t\geq 2$ and we can apply Lemma~\ref{lemm:subtree-no-listener}. Thus, there is a set $X$ of at most $63k^2$ vertices that may contain a watcher ($X$ is computable in linear time). Then, it suffices to try all the possibilities of selecting $k-|W|-1$ additional watchers from this set. There are at most ${63k^2 \choose k}\in O\left(2^{O(k\log k)}\right)$ such possibilities. For each of them, we check in linear time whether we have a valid detection pair. In total this phase takes $2^{O(k\log k)}n^2$ time. If we find a valid solution in this phase, we return YES.

\medskip

Next, we assume that $|L|\geq 2$ and proceed in three phases.

\medskip

\noindent\textbf{Phase 1: Handling the special branching points.} We now do a preprocessing step using Lemma~\ref{lemma:SBP}. First compute (in linear time) the set of special branching points of $T$, and let $B_{2^+}$ be the set of those special branching points that have at least two legs attached. Note that by Observation~\ref{obs:specialBP}, if $|B_{2^+}| > k$ we can return NO. Hence, assume that $|B_{2^+}| \leq k$. Now, for each special branching point $x$ of $B_{2^+}$, by Lemma~\ref{lemma:SBP} we can assume that there are at most four different choices for the set of detectors on $x$ and the vertices belonging to a leg attached to $x$. Therefore, we may go through each of the possible combinations of these choices; there are at most $4^k$ of them. Of course we discard the choices for which there are more than $k$ detectors. For a combination $C_i$, Let $(W_i,L_i)$ be the partial detection pair corresponding to $C_i$. In the remainder, we show how to decide whether there is a detection pair $(W,L)$ of $T$ with $W_i\subseteq W$, $L_i\subseteq L$ and $||(W,L)||\leq k$.

\medskip

\noindent\textbf{Phase 2: determining the set of listeners.} Let $V_i^*$ be the set of vertices of $T$ containing, for each pair $u,v$ of listeners of $L_i$, all the vertices of the path $P_{uv}$ connecting $u$ and $v$, as well as all the vertices of each leg of $T$ attached to a degree~$3$ inner-vertex of $P_{uv}$. By Lemma~\ref{lemm:between2listeners}, all the vertices in $V_i^*$ are distinguished by $(W_i,L_i)$. Clearly, the subgraph $T_i^*=T[V_i^*]$ of $T$ is a connected subtree. Let $X=\{x_1,\ldots,x_p\}$ be the set of vertices in $V_i^*$ that have a neighbour in $T-T_i^*$. For each vertex $x_j$ of $X$, we let $T_{x_j}$ be the subtree of $T$ formed by $x_j$ together with all trees of $T-T_i^*$ containing a neighbour of $x_j$. Let $\mathcal F$ be the forest consisting of all trees $T_{x_j}$ ($1\leq j\leq p$) such that at least one vertex of $T_{x_j}$ is not distinguished by $(W_i,L_i)$. Next, we upper-bound the size of $\mathcal F$.

\begin{claim}\label{clm:mathcal F}
If there is a detection pair $(W,L)$ of $T$ with $W_i\subseteq W$, $L_i\subseteq L$ and $||(W,L)||\leq k$, then $|\mathcal F|\leq k$.
\end{claim}
\claimproof We prove that each tree of $\mathcal F$ must contain a detector of $(W,L)$. Assume for contradiction that there is a tree $T_{x_a}$ of $\mathcal F$ such that no detector of $(W,L)$ belongs to $T_{x_a}$. By definition, $T_{x_a}$ contains a vertex $w$ not distinguished by $(W_i,L_i)$, hence there is a vertex $w'$ of $T-T_i^*$ (say, $w'\in T_{x_b}$) not separated by any listener of $L_i$ and none of $w$, $w'$ is dominated by a watcher of $W_i$. Then, $w'$ must belong to $T_{x_a}$ (hence $a=b$), indeed by definition of $T_i^*$ there are two listeners $u$ and $v$ of $L_i$ whose path contains $x_a$ and $x_b$. But then, if $a\neq b$, clearly $u$ and $v$ would separate $w$ and $w'$, a contradiction. But since $w$ and $w'$ both belong to $T_{x_a}$, no listener outside of $T_{x_a}$ can separate them, and clearly no watcher outside of $T_{x_a}$ can dominate them, which proves the claim.\smallqed

\medskip

By Claim~\ref{clm:mathcal F}, if we have $|\mathcal F|>k$, we can discard the combination $C_i$; hence, we assume that $|\mathcal F|\leq k$.

Now, from each tree $T_{x_j}$ of $\mathcal F$, we build the subtree $\widehat{T_{x_j}}$ of $T_{x_j}$ as follows: (i) remove all legs in $T_{x_j}$ whose closest special branching point has only one leg attached; (ii) remove each leaf of $T_{x_j}$ that is adjacent to a watcher of $W_i$. We also denote by $\widehat{\mathcal F}$ the forest containing each tree $\widehat{T_{x_a}}$ with $T_{x_a}\in\mathcal F$. We next claim the following.

\begin{claim}\label{clm:widehat(T_{x_a})}
If there is a detection pair $(W,L)$ of $T$ with $W_i\subseteq W$, $L_i\subseteq L$ and $k_a$ detectors on $T_{x_a}$, then $\widehat{T_{x_a}}$ has at most $2k_a$ leaves.
\end{claim}
\claimproof Let us consider the legs of $\widehat{T_{x_a}}$ containing no detector of $(W_i,L_i)$. By~(i) of the construction of $\widehat{T_{x_a}}$, such a leg must be attached to a special branching point having at least two legs attached. (Note that operation (i) did not create any new special branching point.) Moreover, by~(ii), if this leg has length~$1$ in $\widehat{T_{x_a}}$, its unique vertex will not be dominated by a watcher, since the set of detectors on this leg and its special branching point are determined by the combination $C_i$. Note that in~(ii), if we create a new leg, this leg contains a detector (a watcher).

Let $v$ be a special branching point of $\widehat{T_{x_a}}$ and let $\mathcal L_v$ be the set of legs in $\widehat{T_{x_a}}$ attached to $v$ that do not contain any detector of $(W_i,L_i)$. By the previous discussion, we have $|\mathcal L_v|\geq 2$. If some leg of $\mathcal L_v$ has length~$1$, as we observed there will be no watcher on $v$. Hence, in order to separate all the neighbours of $v$ on the legs of $\mathcal L_v$, we need a detector in at least $|\mathcal L_v|-1\geq |\mathcal L_v|/2$ of these legs. Otherwise, we still need that many detectors in order to separate all the vertices at distance~$2$ of $v$ on the legs of $\mathcal L_v$. Hence, at most half of the leaves of $\widehat{T_{x_a}}$ belong to a leg with no detector, which proves the claim.\smallqed

\medskip

By Claim~\ref{clm:widehat(T_{x_a})}, $\widehat{T_{x_j}}$ has at most $2k_j$ leaves (where $k_j$ is the set of detectors placed on $T_{x_j}$). Therefore, if the total number of leaves in the forest $\widehat{\mathcal F}$ is more than $2k$, we can discard the current combination $C_i$. Hence, we can assume that the total number of leaves in $\widehat{\mathcal F}$ is at most $2k$.

Therefore, the total number of vertices of degree at least~$3$ in $\widehat{\mathcal F}$ is at most $2k$. Let us consider the sets of maximal threads in $\widehat{\mathcal F}$ (that is, paths with inner-vertices of degree~$2$). The endpoints of such threads are vertices of degree~$1$ or at least~$3$ in $\widehat{\mathcal F}$, and therefore the number of such threads is at most the number of these vertices, which is upper-bounded by $4k$.

Each such thread $P_{uv}$ between two vertices $u$ and $v$ of a tree $\widehat{T_{x_j}}$ of $\widehat{\mathcal F}$, corresponds (in $T_{x_j}$) to a path whose inner-vertices either have a single leg attached, or contain a watcher of $W_i$ with a set of leaves attached. Hence, we may assume that $P_{uv}$ contains at most two listeners: if it contained more than two, by Lemma~\ref{lemm:between2listeners-coro}, we could replace them with the two endpoints $u$ and $v$ of $P_{uv}$. Moreover, if $P_{uv}$ contained exactly one listener, similarly we may assume it is placed on one of the endpoints of $P_{uv}$. Indeed, since $|L|\geq 2$, there will be another listener in the solution set. If this listener is closer to $u$, similarly as in Lemma~\ref{lemm:between2listeners}, a listener on $v$ would distinguish all vertices on $V(P_{uv}^+)$ (where $V(P_{uv}^+)$ contains all vertices of $P_{uv}$ and the vertices of legs attached to inner-vertices of $P_{uv}$). Moreover, the set of pairs of vertices outside of $V(P_{uv}^+)$ separated by $u$ and by $v$ are the same. Note that by Lemma~\ref{lem:unique-leg}, if we choose to place a listener on an endpoint $u$ of a path $P_{uv}$ that is a thread in $\widehat{T_{x_j}}$ and this endpoint has a leg attached, instead of placing the listener on $u$ we can instead place the listener on the leaf of this leg.

Therefore, for each thread $P_{uv}$ of $\widehat{T_{x_a}}$, there are four possibilities for the placement of the listeners: a listener on both $u$ and $v$ (or the leaf of a corresponding leg attached to $u$ or $v$); a single listener at $u$; a single listener at $v$; no listener. There are therefore at most $4^{4k}$ possibilities to guess the placement of listeners on tree $T_{x_a}$. For each such placement $P_j$, we obtain a set $L_j$ of listeners (of course, if $||(W,L_i\cup L_j)||>k$ we do not consider this placement). Now, we can assume that $L=L_i\cup L_j$, that is, that there will be no more listeners in the sought solution $(W,L)$. Hence, it remains to check whether we can obtain a valid solution of size~$k$ by only adding watchers.

\medskip

\noindent\textbf{Phase 3: Determining the set of watchers.} We compute a tree $T_{i,j}^*$ similarly as $T_i^*$ but using the new set $L_i\cup L_j$ of listeners. Similarly as before, we define the set $X'$ of vertices of $T_{i,j}^*$ having a neighbour in $T-T_{i,j}^*$, and the set of trees $T'_{x_j}$. We also let $\mathcal{F'}$ be the forest containing those trees $T'_{x_j}$ that have at least one vertex not distinguished by $(W_i,L_i\cup L_j)$. We know that for each tree $T'_{x_j}$ of $\mathcal{F'}$, we will use only watchers to distinguish the vertices not yet distinguished. For each tree $T'_{x_j}$ in $\mathcal{F'}$, let $\widetilde{T'_{x_j}}$ be the tree obtained from $T'_{x_j}$ by removing all leaves already dominated by a watcher of $W_i$. If this process has created a $t$-stem with $t\geq 2$, then at least $t-1$ leaf neighbours of this $t$-stem have a watcher of $W_i$; simply remove them. Then, the obtained tree $\widetilde{T'_{x_j}}$ has no $t$-stems with $t\geq 2$. Hence, we can apply Lemma~\ref{lemm:subtree-no-listener} to $\widetilde{T'_{x_j}}$ (if necessary, we augment $\widetilde{T'_{x_j}}$ to satisfy the maximality condition~(iii) in Lemma~\ref{lemm:subtree-no-listener}). Thus, there is a set $X_j$ of at most $63k^2$ vertices of $\widetilde{T'_{x_j}}$ that may possibly contain a watcher, and $X_j$ can be computed in time linear in $|V(\widetilde{T'_{x_j}})|$. As before, there are at most $k$ trees in $\mathcal{F'}$ and hence in total the set $X=\cup_j X_j$ of possible placements for the remaining watchers has size at most $63k^3$ (and is computable in linear time). Therefore, we can check all the ${63k^3 \choose k}\in O\left(2^{O(k\log k)}\right)$ possibilities of placing the remaining watchers. If we find a valid detection pair $(W,L)$ with $W_i\subseteq W$, we return YES; otherwise, we discard the placement $P_j$ and move on to the next possibility.


\medskip

It is clear, by the lemmas used in the description in the algorithm, that this algorithm is correct. The running time of the preliminary phase is $2^{O(k\log k)}n^2$. The running time of Phases~1 to~3 is at most $4^k4^{4k}2^{O(k\log k)}n$ which is $2^{O(k\log k)}n$. This completes the proof.
\end{proof}

\section{Conclusion}\label{sec:conclu}

We have obtained a linear-time $2$-approximation algorithm for \DPopt on trees, but perhaps this algorithmic upper bound could be improved to a PTAS? It seems that the reduction of~\cite{FHY15} does not show the inapproximability of \DPopt for any constant greater than one. Therefore, it remains to settle the exact approximation complexity of \DPopt on trees.

As a second question, can the factor $2^{O(k\log k)}$ in our FPT algorithm for \DPdec be improved to single-exponential FPT time $2^{O(k)}$ (or even sub-exponential time $2^{o(k)}$)? For this, the bottleneck is our use of Lemma~\ref{lemm:subtree-no-listener} in the preliminary phase and in Phase~3 of the algorithm, during which we search through all possible subsets of size~$k$ of some sets of size $O(k^2)$ and $O(k^3)$, respectively. Moreover, it is perhaps possible to obtain a linear running-time (in terms of $n$) by reducing the running-time of the preliminary phase.

As an extension of our results, it would be of interest to determine whether, on further graph classes, \DPopt is constant-factor approximable and whether \DPdec is FPT. One natural research direction is to consider the class of planar graphs or its subclasses, the outerplanar graphs and the series-parallel graphs (that is, graphs of treewidth~$2$). These three classes contain all trees. Although both questions are settled in the affirmative for \DSopt and \DSdec (see~\cite{dom:advancedtopics} and~\cite{DF13,N06}), as far as we know they also remain open for \MDopt and \MDdec (though, a polynomial-time algorithm exists for \MDopt on outerplanar graphs~\cite{DPSV12}).

\medskip

\paragraph{Acknowledgements.} We thank the anonymous referees for their careful reading and valuable comments. We wish to thank Bert L. Hartnell for sharing the manuscript~\cite{FHY15} with us. Moreover we acknowledge the financial support of the programme ''IdEx Bordeaux -- CPU (ANR-10-IDEX-03-02)''.

\end{document}